\documentclass[11pt]{article} \textwidth 6.5in\textheight 51pc

\usepackage{amsthm}
\usepackage{amsmath}
\usepackage{amssymb}
\usepackage{hyperref}
\usepackage{natbib}
\usepackage{graphicx}
\usepackage{caption}
\usepackage{verbatim}

\newtheorem{lemma}{Lemma}[section]
\newtheorem{theorem}{Theorem}
\newtheorem{proposition}[lemma]{Proposition}

\theoremstyle{remark}

\DeclareMathOperator{\Tr}{Tr}

\newcommand{\lea}{<^+}
\newcommand{\gea}{>^+}
\newcommand{\eqa}{=^+}

\newcommand{\lel}{<^{\log}}
\newcommand{\gel}{>^{\log}}
\newcommand{\eql}{=^{\log}}

\newcommand{\lem}{\stackrel{\ast}{<}}
\newcommand{\gem}{\stackrel{\ast}{>}}
\newcommand{\eqm}{\stackrel{\ast}{=}}

\newcommand{\tr}{\operatorname{Tr}}

\newcommand{\bmu}{\boldsymbol{\mu}}

\newcommand{\m}{\mathbf{m}}

\newcommand{\QC}{\text{\rmfamily\mdseries\upshape QC}}

\newcommand{\Hg}{\mathbf {Hg}}

\newcommand{\Hv}{\mathbf {Hv}}

\newcommand{\I}{\mathbf I}

\newcommand\floor[1]{{\lfloor#1\rfloor}}\newcommand\ceil[1]{{\lceil#1\rceil}}
\newcommand\FS{\{0,1\}^*}
\newcommand\IS{\{0,1\}^\infty}
\newcommand\FIS{\{0,1\}^{*\infty}}
\newcommand\R{\mathbb{R}}
\newcommand\Q{\mathbb{Q}}
\newcommand\N{\mathbb{N}}
\newcommand\W{\mathbb{W}}

\newcommand\BT{\Sigma}
\renewcommand{\i}{{\mathrm{i}}}

\newcommand\K{{\mathbf K}}
\newcommand\C{{\mathbb C}}
\renewcommand\d{{\mathbf d}}

\newcommand\Ks{\chi} 
\newcommand\ch{{\mathcal H}}

\newcommand\bra[1]{{\langle #1|}}
\newcommand\ket[1]{{| #1 \rangle}}
\newcommand\braket[2]{{\langle #1 | #2 \rangle}}

\newcommand{\kpsi}{\ket{\psi}}
\newcommand{\ktheta}{\ket{\theta}}
\newcommand{\dom}{\mathrm{Dom}}
\newcommand{\supp}{\mathrm{Supp}}

\begin{document}
	
\author {{\sc Samuel Epstein}\footnote{JP Theory Group. \href{mailto:samepst@jptheorygroup.org}{samepst@jptheorygroup.org}
}}
\title{ \textbf{An Extended Coding Theorem with Application to Quantum Complexities}}\date{\today}\maketitle

\begin{abstract}
	This paper introduces a new inequality in algorithmic information theory that can be seen as an extended coding theorem. This inequality has applications in new bounds between quantum complexity measures.
\end{abstract}
\section{Introduction}

In \citealt*{EpsteinLe14}, a new inequality in the field of algorithmic information theory was proven. For a finite set of natural numbers $D$, it was shown that the size of the smallest description of an element of $D$, $\min_{x\in D}\K(x)$, is not much smaller than the negative logarithm of the algorithmic probability of the set, $-\log\sum_{x\in D}\m(x)$. This inequality holds for non-exotic sets whose encoding has little mutual information with the halting sequence, $\I(D;\mathcal{H})=\K(D)-\K(D|\mathcal{H})$.
\begin{align*}
\min_{x\in D}\K(x) &\lel -\log\sum_{x\in D}\m(x) + \I(D;\mathcal{H}).
\end{align*}
Due to algorithmic conservation laws, there are no algorithmic means to produce sets with arbitrary high mutual information with the halting sequence. In this paper, we introduce an update on the above inequality, proving for non-exotic maps $f$ between whole numbers with a finite domain, $\min_{x\in\mathrm{Dom}(f)}\K(x)+f(x)$ is close to the amount $-\log \sum_{x\in\mathrm{Dom}(f)}\m(x)2^{-f(x)}$. Exotic maps $f$ have encodings with high mutual information with the halting sequence, $\I(f;\mathcal{H})$, with
\begin{align*}
\min_{x\in \mathrm{Dom}(f)}\K(x) +f(x)&\lel -\log\sum_{x\in \mathrm{Dom}(f)}\m(x)2^{-f(x)} + \I(f;\mathcal{H}).
\end{align*}
The above inequality can be seen as an extended coding theorem. It has applications in proving tighter bounds between algorithmic quantum complexities. In particular, we show that for a non-exotic pure quantum state $\ket{\psi}$, its Vit\'{a}nyi complexity, $\Hv$, is smaller than its G\'{a}cs complexity, $\Hg$, with
\begin{align*}
\Hv(\ket{\psi}) &\lel \Hg(\ket{\psi}) + \I(\ket{\psi}:\mathcal{H})
\end{align*}
The term $\I(\ket{\psi}:\mathcal{H})$ represents the amount of mutual information between two infinite sequences; one infinite sequence is the halting problem, the other infinite sequence is an infinite encoding of the quantum state $\ket{\psi}$.

Finally this paper contains a section relating quantum Kolmogorov complexity and G\'{a}cs entropy, resolving, in part, open issue (1) of \cite{Gacs01}.
\section{Related Work}

The study of Kolmogorov complexity originated from the work of~[\citealt*{Kolmogorov65}]. The canonical self-delimiting form of Kolmogorov complexity was introduced in~[\citealt*{Levin74}] and~[\citealt*{Chaitin75}]. The universal probability $\m$ was introduced in~[\citealt*{Solomonoff64}]. More information about the history of the concepts used in this paper can be found in the textbook~[\citealt*{LiVi08}]. The main theorem of this paper is an inequality that has the mutual information of a string with the halting sequence. More background on this term can be found in~[\citealt*{VereshchaginVi04v2}]. Lemma \ref{lem:elemmap}, uses the notion of stochasticity. A string is stochastic if it is typical of a simple probability distribution. Aspects involving stochastic objects were studied in~[\citealt*{Shen83,Shen99,Vyugin87}]. Stochasticity is one area of study in algortihmic statistics, which can be found in~[\citealt*{VereshchaginVi04,VereshchaginVi10,Vereshchagin13,VereshchaginSh16}].  
\section{Conventions}
\label{sec:conv}
We use $\R$, $\N$, $\W$, $\C$, $\FS$, and $\IS$ to represent reals, natural numbers, whole numbers, complex numbers, finite strings, and infinite strings. Let $X_{\geq 0}$ and $X_{>0}$ be the sets of non-negative and of positive elements of $X$. The length of a string $x{\in}\{0,1\}^n$ is denoted by $\|x\|=n$. The removal of the last bit of a string is denoted by $(p0^-){=}(p1^-){=}p$, for $p\in\{0,1\}^*$. For the empty string $\emptyset$, $(\emptyset^-)$ is undefined. We use $\FIS$ to denote $\FS{\cup}\IS$, the set of finite and infinite strings.  The $i$th bit of a string $x\in\FIS$ is denoted by $x[i]$. The first $n$ bits of a string $x\in\FIS$ is denoted by $x[0..n]$. The indicator function of a mathematical statement $A$ is denoted by $[A]$, where if $A$ is true then $[A]=1$, otherwise $[A]=0$. The size of a finite set $S$ is denoted to be $|S|$. As is typical of the field of algorithmic information theory, the theorems in this paper are relative to a fixed universal  machine, and therefore their statements are only relative up to additive and logarithmic precision.

For positive real functions $f$ the terms  ${\lea}f$, ${\gea}f$, ${\eqa}f$ represent ${<}f{+}O(1)$, ${>}f{-}O(1)$, and ${=}f{\pm}O(1)$, respectively. In addition ${\lem}f$, ${\gem}f$, and ${\eqm}$ denote $<f/O(1)$, $>f/O(1)$ and $=f*/O(1)$, respectively. For nonnegative real function $f$, the terms ${\lel}f$, ${\gel} f$, ${\eql}f$ represent the terms ${<}f{+}O(\log(f{+}1))$, ${>}f{-}O(\log(f{+}1))$, and ${=}f{\pm}O(\log(f{+}1))$, respectively. A discrete measure is a nonnegative function $Q:\W\rightarrow \R_{\geq 0}$ over whole numbers. The support of a measure $Q$ is the set of all elements $a\in\W$ that have positive measure, with $\supp(Q) = \{a\,{:}\,Q(a)>0\}$. The mean of a function $f:\W\rightarrow\R$ by a measure $Q$ is denoted by $\mathbf{E}_Q[f] = \sum_{a\in\W} f(a)Q(a)$. We say measure $Q$ is a semimeasure iff $\mathbf{E}_Q[1]\,{\leq}\,1$. Furthermore, we say that measure $Q$ is probability measure iff $\mathbf{E}_Q[1]\,{=}\,1$. The image of a measure $Q$ with respect to a (partial) function $f:\W\rightarrow\W$ is defined to be $(fQ)(x) = \sum\,\{\,Q(y)\,{:}\,f(y)\,{=}\,x, y\,{\in}\,\W\}$. If $Q$ is a semimeasure then $fQ$ is also semimeasure (and analogously for probability measures and total functions). 
\subsection{Self Delimiting Codes}
\label{subsec:selfdelimit}
The prefix operator over two strings is denoted by $\sqsubseteq$, where for finite string $x\in \FS$ and arbitrary string $y\in \FIS$, we say $x\sqsubseteq y$ iff there exists some $z\in\FIS$ such that $xz=y$. Furthermore, $x\sqsubseteq x$ for $x\in \IS$. When it is clear from the context, we will use whole numbers and other finite objects interchangeably with their binary representations. For example, each whole number $n\in\W$ can be associated with the $(n\,{+}\,1)$th string of a length increasing lexicographical ordering $\{\xi_n\}_{n=1}^\infty$, $\xi_n\in\FS$, with
\begin{align*}
(0,0), (1,1), (2,00), (3,01), (4,10), (5,11), (6,000)\dots
\end{align*}

Thus $\xi_6 = 000$. A prefix free set of of codes $S\subset \FS$ is a set of strings such that there does not exist two distinct strings $x,y$ in $S$ where one string is a prefix of the other, $x\sqsubseteq y$. By the Kraft inequality, for such a prefix free set $S$ of strings,
$$\sum_{x\in S}2^{-\|x\|}\leq 1.$$
We say such $S$ is a \textit{self-delimiting} code because there exists a method to determine where each code word $x\in S$ ends without reading past its last symbol. One such code word is $\langle x\rangle' = 1^{\|x\|}0x$, where the decoding algorithm would first count the number of $1's$ before the first 0 to determine the length of $x$ and then output the $\|x\|$ remaining bits in the input, (corresponding to $x$). Thus $\|\langle x\rangle'\| = 2\|x\|+1$. For a finite string $x$, we use $\langle x\rangle\in\FS$ to denote a more efficient self-delimiting code, with $\langle x\rangle = \langle \xi_{\|x\|}\rangle'0x$. For example, $\langle 11111\rangle = 11011011111$. Thus $\|\langle x\rangle\| \leq \|x\|+2\ceil{\log \|x\|}+2$.

The encoding of a finite set $\{x_n\}_{n=1}^m$ of strings (natural numbers) is defined to be $\langle\{x_1,\dots,x_m\}\rangle=\langle m\rangle\langle x_1\rangle\dots\langle x_m\rangle$, and it is also denoted as $\langle x_1,\dots,x_m\rangle$. The encoding of a rational number $r\in\Q$ is defined to be $\langle r\rangle = \langle p,q\rangle$ for reduced $p/q=r$. $\langle x,\alpha\rangle = \langle x\rangle\alpha$ for $x\in\FS$ and $\alpha\in\IS$. For two infinite strings $\alpha,\beta\in\IS$, their encoding is $\langle 2\rangle\langle \alpha,\beta \rangle =\alpha_1\beta_1\alpha_2\beta_2\alpha_3\beta_3\dots$. For sequences $\alpha$, $\beta$, and $\gamma$, we have that $\langle \alpha,\beta,\gamma\rangle=\langle 3\rangle\langle \langle \alpha,\beta\rangle,\gamma\rangle$, and similarly for for $n$ sequences. The encoding of a real $r$, $\langle r\rangle$, is the encoding of two infinite sequences of all rational number that are smaller and bigger than $r$, respectively. The terms $\mathrm{Dom}(F)$ and $\mathrm{Range}(F)$ denote the domain and range of a function $F$. If a function $F$ from whole numbers to whole numbers has a finite domain, then its self delimiting code is denoted by $\langle F\rangle= \langle \{( a,F(a))\;:\;a\in\mathrm{Dom}(F)\}\rangle$. We call such functions, {\em elementary maps}. A (semi) measure $Q:\W\rightarrow \Q_{\geq 0}$ with a finite support and a range of nonnegative rational numbers is called a {\em elementary (semi) measure}, and its self delimiting code is denoted by $\langle Q\rangle= \langle \{\langle  a,Q(a)\rangle\;:\;a\in\supp(Q)\}\rangle$. Both elementary maps and measures admit a finite explicit descriptions.

\subsection{Algorithms.}
Our paper uses self-delimiting machines $M$ which have four tapes: a main input tape, an auxiliary input tape, a work tape, and an output tape. The alphabet for all tapes is $\{0,1,\$\}$. We say that $M$ computes the partial function $T:\FS\times\FIS\rightarrow\FS$,  if whenever $T_\alpha(x)$ is defined, $M$ outputs $y=T_\alpha(x)$ when given $x\in\FS$ and $\alpha\in\FIS$ as input. More specifically, 
\begin{enumerate}
	\item  $M$ starts with all its heads in the leftmost square. The main input tape starts with $x\$^\infty$. The auxiliary tape is set to $\alpha$ if it is an infinite string, otherwise it starts with $\alpha\$^\infty$. The work and output tape start with $\$^\infty$.
	\item During its operation, $M$ reads exactly $\|x\|$ bits from the main input tape.
	\item The output tape is $y\$^\infty$ when $M$ halts.
\end{enumerate}
When inputs $x$ and $\alpha$ are not defined for $T$, we say $T_\alpha(x)=\perp$ and the machine $M$ does not perform the steps enumerated above. A partial function $T$ is self delimiting, or prefix free, if there is a self delimiting machine that computes it. The domain of such $T$ is prefix free, where for all $x,y\in\FS$, $\alpha\in\FIS$, with $y\,{\neq}\,\emptyset$, it must be that $T_\alpha(x)\,{=}\perp$ or $T_\alpha(xy)\,{=}\perp$.  For convenience, we use symbols to interchangeably to denote both the machines and also the partial functions (between finite strings) they compute.

We use a fixed universal prefix-free machine $U$, where for each prefix-free machine $T$, there exists $t\in\FS$ where $U_\alpha(tx) = T_\alpha(x)$ for all $x\in\FS$ and $\alpha\in\FIS$. One example is for such $t$ to be equal to $\langle i\rangle$, where $i$ is the first index of $T$ in an enumeration of self-delimiting machines. A set of whole numbers is (recursively) enumerable if it is the range of a partial recursive function. We say that a real valued function $f:\W\rightarrow \R$ over whole numbers is upper semi-computable if the set $\{(a,q)\,{:}\,f(a)<q\in\Q\}$ is enumerable. We say $f$ is lower semi-computable if $-f$ is upper semi-computable. 
\subsection{Left-Total Machines}
The notion of the ``left-total'' universal algorithm is needed for the proof of both the mixed state and pure state coding theorems. We say $x\in\FS$ is total with respect to a machine if the machine halts on all sufficiently long extensions of $x$. More formally, $x$ is total with respect to $T_y$ for some $y\in\FIS$ iff there exists a finite prefix free set of strings $Z\subset\FS$ where $\sum_{z\in Z}2^{-\|z\|}=1$ and $T_y(xz)\neq\perp$ for all $z\in Z$.  We say (finite or infinite) string $\alpha\in\FIS$ is to the ``left'' of $\beta\in\FIS$, and use the notation $\alpha\lhd \beta$, if there exists a $x\in\FS$ such that $x0\,{\sqsubseteq}\, \alpha$ and $x1\,{\sqsubseteq}\, \beta$. A machine $T$ is left-total if for all auxiliary strings $\alpha\in\FIS$ and for all $x,y\in\FS$ with $x\lhd y$, one has that $T_\alpha(y)\neq\perp$ implies that $x$ is total with respect to $T_\alpha$. An example can be seen in Figure \ref{fig:LeftTotal}.

\begin{figure}[h!]
	\begin{center}
		\includegraphics[width=0.4\columnwidth]{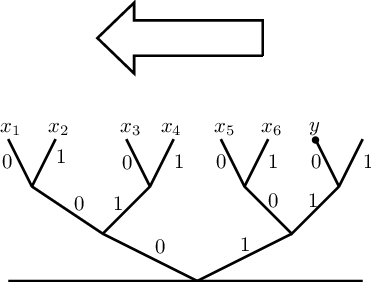}
		\caption{The above diagram represents the domain of a left total machine $T$ with the 0 bits branching to the left and the 1 bits branching to the right. For $i\in \{1..5\}$, $x_i\lhd x_{i+1}$ and $x_i\lhd y$. Assuming $T(y)$ halts, each $x_i$ is total. This also implies each $x_i^-$ is total as well.}
		\label{fig:LeftTotal}
	\end{center}
\end{figure}

For the remaining part of this paper, we can and will change the universal self delimiting machine $U$ into a universal left-total machine $U'$ by the following definition. The algorithm $U'$ enumerates all strings $p\,{\in}\,\FS$ in order of their convergence time of $U(p)$ and successively assigns them consecutive intervals $i_p{\subset}[0,1]$ of width $2^{-\|p\|}$. Then $U'$ outputs $U(p)$ on input $p'$ if the open interval corresponding to $p'$ and not that of $(p')^{-}$ is strictly contained in $i_p$. The open interval in [0,1] corresponding with $p'$ is $([p']2^{-\|p'\|},([p']{+}1)2^{-\|p'\|})$ where $[p]$ is the value of $p$ in binary. For example, the value of both strings 011 and 0011 is 3. The value of 0100 is 4. The same definition applies for the machines $U'_\alpha$ and $U_\alpha$, over all $\alpha\,{\in}\,\FIS$. We now set $U$ to equal $U'$.

\begin{figure}[h!]
	\begin{center}
		\includegraphics[width=0.4\columnwidth]{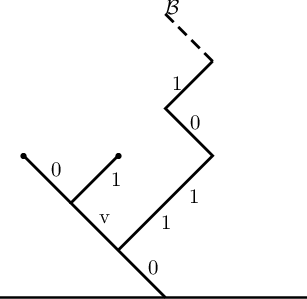}
		\caption{The above diagram represents the domain of the universal left-total algorithm $U$, with the 0 bits branching to the left and the 1 bits branching to the right. The strings in the above diagram, $0v0$ and $0v1$, are halting inputs to $U$ with $U(0v0)\neq \perp$ and $U(0v1)\neq \perp$. So $0v$ is a total string. The infinite border sequence $\mathcal{B}\in\IS$ represents the unique infinite sequence such that all its finite prefixes have total and non total extensions. All finite strings branching to the right of $\mathcal{B}$ will cause $U$ to diverge.}
		\label{fig:DomainUPrime}
	\end{center}
\end{figure}

Without loss of generality, the complexity terms of Section~\ref{sec:complexity} are defined in this section with respect to the universal left total machine $U$. The infinite border sequence $\mathcal{B}\in\IS$ represents the unique infinite sequence such that all its finite prefixes have total and non total extensions. The term ``border'' is used because for any string $x\in\FS$, $x\lhd\mathcal{B}$ implies that $x$ total with respect to $U$ and $\mathcal{B}\lhd x$ implies that $U$ will never halt when given $x$ as an initial input. Figure~\ref{fig:DomainUPrime} shows the domain of $U$ with respect to $\mathcal{B}$. The border sequence is computable from $\mathcal{H}$. 

For all total strings $b\in\FS$, we define the semimeasure $\m_b(x) = \sum \{2^{-\|p\|}\,{:}\,U(p)=x,\;p\lhd b\textrm{ or } b\sqsubseteq p\}$. If $b$ is not total then $\m_b(x)=\perp$ is undefined. Thus the algorithmic weight $\m_b$ of a string $x$ is approximated using programs that either extend $b$ or are to the left of $b$. \newpage
\subsection{Complexity}
\label{sec:complexity}
The Kolmogorov complexity of string $x\,{\in}\,\FS$ relative to string $\alpha\,{\in}\,\FIS$ is defined to be $\K(x|\alpha) = \min \{\|p\|\,{:} \,U_\alpha(p)=x\}$. We will use the term algorithmic entropy interchangeably with Kolmogorov complexity. The Solomonoff prior is defined to be the probability that $U$ outputs $x$ when given random bits as input, with $\m(x|\alpha) = \sum_{p\,{:}\,U_\alpha(p)=x}2^{-\|p\|}$. 

The function $\m$ is a universal lower semi-computable semi-measure, in that for any lower semi-computable semi-measure $p$, there exists $c_p\in\N$ such that for all $x\in\FS$, one has that $p(x) <c_p\m(x)$. By the coding theorem, for all $x\,{\in}\,\FS$, one gets that $\K(x) \eqa -\log\m(x)$. The halting sequence $\ch\in\IS$ is the characteristic sequence of the domain of $U$. Chaitin's omega $\Omega\in\R$ is the probability that the universal machine $U$ will halt when given random bits as input, with $\Omega = \sum_{a\in \W} \m(a)$. The mutual information of two strings $x$ and $y$, conditional to $\alpha\in\FIS$ is $\I(x\,{:}\,y|\alpha)\,{=}\,\K(x|\alpha)\,{+}\,\K(y|\alpha)\,{-}\,\K(x,y|\alpha)$. The information that the halting sequence $\ch$ has about $x$ (or any finite object it encodes), conditional to $\alpha\in\FIS$ is denoted by $\I(x;\ch|\alpha)=\K(x|\alpha){-}\K(x|\alpha,\ch)$. This is the difference of the algorithmic entropy of $x$ and the algorithmic entropy $x$ relativized to $\ch$ (with both complexities also relativized to $\alpha$). 
We define information with respect to two infinite sequences $\alpha,\beta\in\IS$, conditional to $\xi\in\FIS$, as $\I(\alpha\,{:}\,\beta|\xi)=\log\sum_{x,y\in\FS}\m(x|\alpha,\xi)\m(y|\beta,\xi)2^{\I(x\,{:}\,y|\xi)}$, introduced in [\citealt*{Levin74}]. 

This information function follows algorithmic conservation laws. The non-probabilistic conservation law states that for partial recursive $f:\IS\rightarrow\IS$, $\alpha,\beta\in\IS$, $\I(f(\alpha):\beta)\lea\I(\alpha:\beta)+\K(f)$, where $\K(f)$ is the size of the smallest $U$ program $p\in\FS$ which computes $f$, where $U(p\alpha)=p(\alpha)$. This inequality and the probabilistic conservation inequality of $\I$ were stated in [\citealt*{Levin74}], and their proofs can be found in [\citealt*{Vereshchagin19}].

\section{Extended Coding Theorem}

The main theorem of this paper is the extended coding theorem, which states that for an elementary map $f:\W\rightarrow\W$, with domain $D$, $\min_{a\in D}f(a)+\K(a)$ is close to $-\log \sum_{a\in D}\m(a)2^{-f(a)}$, for maps that have low mutual information with the halting sequence.\\

\noindent\textbf{Theorem.}\\ \textit{For elementary map $f$,
	$\min_{a\in \dom(f)} f(a)+\K(a)\lel -\log \sum_{a\in\dom(f)}\m(a)2^{-f(a)} +\I( \langle f\rangle;\mathcal{H})$.\\
	\label{thr:MainTheoremAlgStats}}

\noindent \textit{Proof Sketch.} The proof of this theorem invokes Lemma \ref{lem:elemmap} which is a statement similar to Theorem \ref{thr:MainTheoremAlgStats} except it uses an elementary semi measure and a stochasticity term instead of mutual information with the halting sequence. The elementary semi measure used is $\m_b$ for a total string $b$. The proof invokes Lemma \ref{lemmaStochastic} to replace the stochasticity term with mutual information with the halting sequence. Then conditioning of terms on $b$ is removed.\\ 

The deficiency of randomness of $x\in\FS$ with respect to probability measure $Q$, conditional to $v\in\FS$, is $\d(x|Q,v)=\floor{\,-\log Q(x)\,}{-}\K(x|v)$, and we say 
$x$ is typical of $Q$, conditional to $v$, if $\d(x|Q,v)=O(1)$. We say that $x$ is stochastic if it is typical of some simple probability measure. Stochasticity is a fundamental part of algorithmic statistics, and a comprehensive survery of this topic can be found in [\citealt*{VereshchaginSh16}].

More formally, one says that $x$ is $(j,k)$ stochastic for $j,k\in\N$ if there exists a elementary probability measure $Q$, with  $\langle Q\rangle=U(v)$, $v= \{0,1\}^j$, and $\d(x|Q,v)\leq k$. We use the stochasticity function $\Ks(x)=\min\{j{+}3\log(k)\,{:}\,x\textrm{ is }(j,k)\textrm{ stochastic}\}$. The conditional stochasticity form\footnote{This is formally represented as $\Ks(x|\alpha) = \min\{j+3\log(k)\,{:}\,\exists v\in\{0,1\}^j, U_\alpha(v)=\langle Q\rangle, \d(x|Q, \langle v,\alpha\rangle)\leq k\}$.}
is represented by $\Ks(x|\alpha)$, for $\alpha\,{\in}\,\FIS$. We recall from Section \ref{sec:conv} that functions between whole numbers are called {\em elementary maps} if they have a finite domain.  Semi measures over whole numbers with finite support and with a range containing only nonnegative rational numbers are called {\em elementary semi measures}. 

\begin{lemma} 
	\label{lem:elemmap}
	Let $f$ be a elementary map and $m$ be a elementary semi measure. Let $a\in\W$ vary over $\dom(f)$. Then $\min_a f(a)+\K(a|m)\lel -\log \sum_a m(a)2^{-f(a)} + \Ks(f|m)$.
\end{lemma}

\begin{proof} If $m$ is not a proper probability measure, and $R$ is the support of $m$, we modify $m$  to give an arbitary $b\in\W$, the value of $1-m(R)$. So $m$ can be assumed to be an elementary probability measure. Since all terms in the theorem are conditioned on $\langle m\rangle$, we will also condition all complexity terms in the proof on $\langle m\rangle$ and drop its notation. More formally, $U(x)$ is used to denote $U_{\langle m\rangle}(x)$, $\K(x)$ is used to denote $\K(x|m)$,  and $\Ks(f)$ is used to denote $\Ks(f|m)$. \\
	
\noindent  For any elementary map $g$, let $g_n = g^{-1}(n)\cap\supp(m)$ and let $g_{\leq n}=\cup_{i=0}^n g_i$, for $n\,{\in}\,{\W}\,{\cup}\,\{\infty\}$. Let $s=\ceil{-\log \sum_{a\in f_{\leq \infty}}m(a)2^{-f(a)}}$. Using the reasoning of Markov's inequality,
	\begin{align}
	\label{eqn:sumLeqS1}
	\sum_{a\in f_{\leq \infty}}m(a)2^{-f(a)} &\geq 2^{-s},\\
	\label{eqn:sumLeqS2}
	\sum_{a\in f_{\leq \infty}{\setminus}f_{\leq s}}m(a)2^{-f(a)} &\leq \sum_{a\in f_{\leq \infty}{\setminus}f_{\leq s}}m(a)2^{-s-1} \leq 2^{-s-1},\\
	\label{eqn:sumLeqS3}
	\sum_{a\in f_{\leq s}}m(a)2^{-f(a)} &\geq 2^{-s-1}.
	\end{align}
	Equation~(\ref{eqn:sumLeqS1}) follows from the definition of $s$ and Equation~(\ref{eqn:sumLeqS3}) follows from Equations~(\ref{eqn:sumLeqS1}) and~(\ref{eqn:sumLeqS2}). We now turn our attention to creating an elementary probability measure $Q$ with the following properties:
	\begin{enumerate}
		\item $f$ is typical of $Q$ and $Q$ is simple, i.e. there is a $v\in\FS$ with $U(v)=\langle Q\rangle$ and $\|v\| +3\log\max\{\d(f|Q,v),1\}$ is not much larger than $\Ks(f)$.
		\item All strings in the support of $Q$ encode elementary functions $g$ whose range contain a lot of values that are not greater than $s$, with $\sum_{a\in g_{\leq s}}m(a)2^{-g(a)} \geq 2^{-s-1}$.
	\end{enumerate}
	
	To accomplish this goal, we start with the program $v'\,{\in}\,\FS$ and elementary probability measure $Q'$ that realizes the stochasticity of $f$, with $U(v')=\langle Q'\rangle$, and also with the relation $\Ks(f)\,{=}\,{\|v'\|}\,{+}3\log\max\{\d(f| Q',v'),1\}$. Note that this implies $\langle f\rangle\in\supp(Q')$. Let $Q$ be the elementary probability measure equal to $Q'$ conditioned on the set of (encoded) elementary maps $g$ such that $\sum_{a\in g_{\leq s}}m(a)2^{-g(a)}\geq 2^{-s-1}$. Thus $Q(\langle g\rangle) = [g\,{\in}\,S]Q'(g)/Q'(S)$, where $S\subset\FS$ , the support of $Q$, is defined as $S = \{\langle g\rangle\,{:}\,g\in\supp(Q'), \sum_{a\in g_{\leq s}}m(a)2^{-g(a)}\geq 2^{-s-1}\}$. This $Q$ is computable from $v'$ and $s$. Using this fact, define the $Q$ program $v\in\FS$, to be of the form $v=v_0v_sv'$, where $v_0\in\FS$ is helper code of size $O(1)$, and $v_s\in\FS$ is a shortest $U$-program for $s$. So $\|v\| \lea \|v'\|+\K(s)$. We define $d=\max\{\d(f|Q,v),1\}$ and we have that
	\begin{align}
	\nonumber
	\|v\|&\lea \|v'\| + \K(s),\\
	\nonumber
	\|v\|+3\log d &\lea \|v'\| + \K(s) + 3\log d\\
	\nonumber
	&\lea \|v'\| + \K(s) + 3\log(\max\{-\log Q(f) - \K(f|v),1\})\\
	\label{eq:StochF0}
	&\lea \|v'\| + \K(s) + 3\log(\max\{-\log Q'(f) - \K(f|v),1\})\\
	\label{eq:StochF1}
	&\lea \|v'\| + \K(s) +3\log(\max\{-\log Q'(f) - \K(f|v')+\K(v|v'),1\})\\
	\label{eq:StochF2}
	&\lea \|v'\| + \K(s) + 3\log(\max\{-\log Q'(f) - \K(f|v')+\K(s),1\})\\
	\nonumber
	&\lel \|v'\| +\K(s)+ 3\log(\max\{-\log Q'(f)-\K(f|v'),1\}),\\
	\label{eq:StochF3}
	\|v\|+3\log d &\lel \Ks(f)+\K(s).
	\end{align}
	Equation~(\ref{eq:StochF0}) follows from $Q(f)\,{=}\,Q'(f)/Q'(\supp(Q))$, and thus $-\log Q(f)\,{\leq}\,-\log Q'(f)$. Equation~(\ref{eq:StochF1}) follows from the inequality $\K(f|v')\lea \K(f|v)+\K(v|v')$. Equation~(\ref{eq:StochF2}) follows from $v$ being computable from $v'$ and $v_s$, and thus $\K(v|v') \lea \K(s)$. \\
	
	We now create a small set of lists of numbers $A$ that will intersect with the range of a large percentage of the support of $Q$. We do so by using the probabilistic method. Let $c\in\N$ be a constant solely dependent on the universal Turing machine $U$ {\tiny }to be determined later. We use an elementary measure $w_n$ over lists $A^n$ of (possibly repeating) whole numbers of size $cd2^{s+1-n}$ where $w_n(A^n)=\prod_{i=1}^{cd2^{s+1-n}}m(A_i^n)$. For a set of $s+1$ lists $A=\{A^n\}_{n=0}^s$, we a measure $w$ over $A$, where $w(A)=\prod_{n=0}^sw_n(A^n)$.
	
	For a set of lists $A$ and elementary function $g$, let $\mathbf{1}(g,A) = 1$ if $g_n\cap A^n =\emptyset$ for all $n\,{\in}\,[0,s]$, and $\mathbf{1}(g,A) = 0$, otherwise. Thus
	\begin{align}
	\nonumber
	\mathbf{E}_{g\sim Q}\mathbf{E}_{A\sim w}[\mathbf{1}(g,A)]&= \sum_g Q(g)\prod_{n=0}^s(1-m(g_n))^{|A^n|}\\
	\label{eq:ee2}
	&\leq \sum_g Q(g)\prod_{n=0}^s \exp\{-|A^n|m(g_n)\}\\
	\nonumber
	&= \sum_g Q(g)\exp\left\{-\sum_{n=0}^s|A^n|m(g_n)\right\}\\
	\nonumber
	&= \sum_g Q(g)\exp\left\{-\sum_{n=0}^scd2^{s+1-n}m(g_n)\right\}\\
	\nonumber
	&= \sum_g Q(g)\exp\left\{-cd2^{s+1}\sum_{n=0}^sm(g_n)2^{-n}\right\}\\
	\label{eq:ee3}
	\mathbf{E}_{g\sim Q}\mathbf{E}_{A\sim \lambda}[\mathbf{1}(g,A)] &\leq \sum_g Q(g)\exp\left\{-cd\right\}=\exp\left\{-cd\right\}.
	\end{align}
	
	\noindent 
	Equation~(\ref{eq:ee2}) follows from the inequality $(1{-}a){\leq}e^{-a}$ over $a\,{\in}\,[0,1]$. Equation~(\ref{eq:ee3}) follows from the definition of the support of $Q$, where $g\in\supp(Q)$ iff $\sum_{a\in g_{\leq s}}m(a)2^{-g(a)}\geq 2^{-s-1}$. By the probability argument, there exists a set of lists $A\,{=}\,\{A^n\}_{n=0}^s$ such that $|A^n|\,{=}\,cd2^{s+1-n}$ and
	\begin{align*}
	\mathbf{E}_{g\sim Q}[\mathbf{1}(g,A)] &\leq \exp\{-cd\}. 
	\end{align*}
	\noindent  There exists a brute force search algorithm that on input $c$, $d$, $v$, outputs $A$. Note that the strings $s$ and $\langle Q\rangle$ are computable from $v$. This algorithm computes all possible sets of lists $A'=\{A'^n\}_{n=0}^s$, $|A'^n|=cd2^{s+1-n}$, $A'^n\subseteq \supp(Q)$ and outputs the first $A'$ such that $\mathbf{E}_{g\sim Q}[\mathbf{1}(g,A')] \leq \exp\{-cd\}$. The existence of such an $A'$ is guaranteed by Equation~(\ref{eq:ee3}). So
	\begin{align}
	\label{eq:CompA}
	\K(A)&\lea \K(c,d,v).
	\end{align}
	
	We now show that there is an $n$ where $f_n\cap A^n\neq\emptyset$. To do so, we show that any function $g$ in the support of $Q$ whose range does not interesct with $A$, i.e.  $\mathbf{1}(g,A)=1$ will have a very high deficiency of randomness with respect to $Q$ and $v$. For all such $g$ and proper choice of $c$ solely dependent on $U$,
	\begin{align}
	\nonumber
	\d(g|Q,v) &= \floor{-\log Q(g)}-\K(g|v)\\
	\label{eq:gTyp1}
	&>-\log Q(g)-(-\log \mathbf{1}(g,A)\floor{e^{cd}}Q(g)+\K(\mathbf{1}(\cdot,A)\floor{e^{cd}}Q(\cdot)|v))-O(1)\\
	\nonumber
	&> cd \log e - \K(\mathbf{1}(\cdot,A)\floor{e^{cd}}Q(\cdot)|v) - O(1)\\
	\nonumber
	&> cd\log -\K(A,c,d|v) - O(1)\\
	\label{eq:gTyp2}
	&> cd \log e -\K(c,d) >d.
	\end{align}
	
	With $c$ being chosen, it is removed from consideration for the rest of the proof, with $c\in\O(1)$. Equation \ref{eq:gTyp1} is due to the fact that for any elementary semimeasure $P$, $\K(x)\lea \K(P)-\log P(x)$. Equation \ref{eq:gTyp2} is due to Equation \ref{eq:CompA}. So $\mathbf{1}(f,A)=0$, otherwise by the above equation, $\d(f|Q,v)>d$, causing a contradiction. So there exists $n\,{\in}\,[0,s]$ with $a\in f_n\cap A^n$ and
	\begin{align}
	\nonumber
	\K(a) &\lea \log |A^n|\,{+}\,\K(A^n)\\
	\nonumber
	&\lea \log |A^n|\,{+}\,\K(A)\,{+}\,\K(A^n|A)\\
	\label{eq:SimpleElement1}
	& \lea (\log d\,{+}\,s\,{-}\,n)\,{+}\,\K(d,v)\,{+}\,\K(n)\\
	\nonumber
	&\eqa \log d\,{+}\, s\,{-}\,f(a)\,{+}\, \K(d,v)\,{+}\,\K(f(a))  \\
	\nonumber
	\K(a)+f(a) &\lea \log d + s + \K(v)+\K(d) +\K(f(a))\\ 
	\label{eq:SimpleElement2}
	\K(a) +f(a) & \lel s\,{+}\, \|v\|\,{+}\,3\log d\\
	\label{eq:SimpleElement3}
	\K(a) +f(a) & \lel s \,{+}\,\chi(f)\\
	\label{eq:SimpleElement4}
	\min_{a\in f_{\leq \infty}}\K(a) \,{+}\,f(a) & \lel {-}\,{\log} {\sum_{a\in f_{\leq \infty}}m(a)2^{-f(a)}}\,{+}\,\chi(f).
	\end{align}
	Equation~(\ref{eq:SimpleElement1}) follows from Equation~(\ref{eq:CompA}), and from $c\in O(1)$. Equation~(\ref{eq:SimpleElement2}) follows from $\K(x)\lel \|x\|$ for $x\in \FS\cup\W$. Equation~(\ref{eq:SimpleElement3}) follows directly from Equation~(\ref{eq:StochF3}). Equation~(\ref{eq:SimpleElement4}) follows from the definition of $s$ and its form proves the theorem.
\end{proof}

\begin{proposition} 
	\label{prp:borderprefix}
	For border prefix $b\sqsubseteq\mathcal{B}$,  $\K(b|\mathcal{H}) \lea \K(\|b\|)$ and $\|b\|\lea\K(b)$.
\end{proposition}
\begin{proof}
	The border $\mathcal{B}$ is computable from the halting sequence $\mathcal{H}$, so it follows easily $\K(b|\mathcal{H})\lea\K(\|b\|)$. We recall that $\Omega =\sum_x\m(x)$ is Chaitin's Omega, the probability that U will halt. It is well known that the binary expansion $\Omega'\in\IS$ of $\Omega$ is Martin L\"{o}f random. Given $b\sqsubset\mathcal{B}$, $\|b\|\in\BT^n$, one can compute $\hat{\Omega}=\sum\{2^{-\|y\|}[U(y)\neq\perp] : y\lhd b\}$ with differs from $\Omega$ in the summation of programs which branch from $\mathcal{B}$ at positions $n+1$ or higher. Thus $\Omega-\hat{\Omega}\leq 2^{-n}$. So $n\lea \K(\Omega'[0..n-1]) \lea \K(\Omega'[0..n-1],b)\lea \K(\Omega'[0..n-1]|b)+\K(b)\lea \K(b)$. 
\end{proof}

\begin{proposition} 
	\label{prp:nontotalprefix}
	If $b\in\FS$ is total and $b^-$ is not total, then $b^-$ is a border prefix, with $b^-\sqsubset\mathcal{B}$. 
\end{proposition}
\begin{proof}
	If $b\in \FS$ is total and $b^-$ is not, then $b^-$ has a total extension $b^-0$ and a non total extension $b^-1$, thus by the definition of the border sequence, $b^-\sqsubset\mathcal{B}$.
\end{proof}
The following lemma shows that non-stochastic strings $x$ are ``exotic,'' {i.e.} have high $\I(x\,{;}\,\ch)$ information with the halting sequence.

\begin{figure}[t]
	\begin{center}
		\includegraphics[width=0.35\columnwidth]{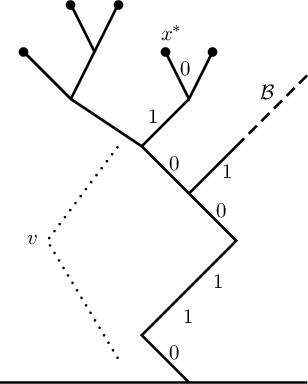}
		\caption{The above figure shows an example of the domain of left-total $U$ with the terms used in Lemma~\ref{lemmaStochastic}. $x^*=0110010$ and $v=01100$. Since $v$ is total and $v^-$ is not, $v^-$ is a prefix of the border sequence $\mathcal{B}$. In the above example, assuming all halting extensions of $v$ produce a unique output, $|\mathrm{Support}(Q)|= 5$, and $Q(x)=2^{-\|x^*\|+\|v\|}=0.25$.}
		\label{fig:LemMain}
	\end{center}
\end{figure}

\begin{lemma}
	\label{lemmaStochastic} 
	For  $x\in\FS$, $\Ks(x)\lel\I(x\,{;}\,\ch)$.
\end{lemma}
\begin{proof}
	Let $U(x^*)=x$, $\|x^*\|=\K(x)$, and $v$ be the shortest total prefix of $x^*$. We define the elementary probability measure $Q$ such that $Q(a) = \sum_w 2^{-\|w\|}[U(vw)\,{=}\,a]$. Thus $Q$ is computable relative to $v$. In addition, since $v\sqsubseteq x^*$, one has the lower bound $Q(x) \geq  2^{-\|x^*\|+\|v\|} = 2^{-\K(x)+\|v\|}$. Therefore 
	\begin{align}
	\nonumber
	\mathbf{d}(x|Q,v) &= \floor{-\log Q(x)} - \K(x|v) \\
	\nonumber
	&\leq  \K(x)-\|v\| - \K(x|v)\\ 
	\nonumber
	&\lea (\K(v) +\K(x|v)) - \|v\| -\K(x|v)\\
	\nonumber
	&\lea (\|v\|+\K(\|v\|)+\K(x|v)) - \|v\| -\K(x|v),\\
	\mathbf{d}(x|Q,v) &\lea \K(\|v\|).\label{eqnDef}
	\end{align}
	
	Since $v$ is total and $v^{-}$ is not total, b\textit{}y proposition~(\ref{prp:nontotalprefix}), $v^-$ is a prefix of the border sequence $\mathcal{B}$ (see~Figure~\ref{fig:LemMain}). In addition, $Q$ is computable from $v$. Therefore
	\begin{align}
	\nonumber
	\K(x|\ch) &\lea \K(x|Q) +\K(Q|\ch)\\
	\nonumber 
	&\lea \K(x|Q) + \K(v|\ch)\\
	&\lea -\log Q(x)+\K(\|v\|)\label{eq:Halting1}\\ 
	\nonumber
	&\lea \K(x) - \|v\| +\K(\|v\|),\\
	\nonumber
	\|v\| &\lea \K(x)-\K(x|\ch) +\K(\|v\|),\\
	\|v\| &\lel \I(x;\ch).\label{eq:Halting2}
	\end{align}
	
	Equation~(\ref{eq:Halting1}) is due to Proposition~(\ref{prp:borderprefix}). Since $Q$ is computable from $v$, one gets $\Ks(x) \lea \K(v)+3\log(\max\{\mathbf{d}(x|Q,v),1\})\lea \|v\|+\K(\|v\|)+3\log(\max\{\mathbf{d}(x|Q,v),1\})$. Due to Equation~\ref{eqnDef}, one gets $\Ks(x) \leq \|v\| + O( \K(\|v\|))\lel \|v\|$. Due to Equation~\ref{eq:Halting2}, one gets $\Ks(x)\lel \I(x;\ch)$.
\end{proof}

\noindent\textbf{Theorem 1.} $ $\\ For elementary map $f$,
	$\min_{a\in \dom(f)} f(a)+\K(a)\lel -\log \sum_{a\in\dom(f)}\m(a)2^{-f(a)} +\I( \langle f\rangle;\mathcal{H})$.
	\label{thr:MainTheoremAlgStats}

\begin{proof} 
	Let $s=\ceil{1-\log \sum_{a\in\dom(f)}\m(a)2^{-f(a)}}$ and let $S(z)=\ceil{-\log \sum_{a\in\dom(f)}\m_z(a)2^{-f(a)}}$ be a partial recursive function from strings to rational numbers. $S$ is defined solely on total strings, where $S(z)\neq\perp$ iff $z$ is total. For total strings $z$, $z^-$, one has that $\m_{z^-}(x)\,{\geq}\,\m_{z}(x)$ and therefore $S(z^-) \leq S(z)$. Let $b$ be the shortest total string with the property that $S(b)\,{<}\,s$. This implies $S(b^-)=\perp$ and thus $b^-$ is not total. So by proposition~(\ref{prp:nontotalprefix}), $b^-\sqsubseteq\mathcal{B}$ is a prefix of border.  Lemma~\ref{lem:elemmap}, with $U$ containing $b$ on an auxilliary tape, with $m(a) = \m_b(a)$, provides $a\in \W$ such that $\K(a|m,b)+f(a) \lel s + \Ks(f|m,b)$. Since $\K(m|b)=O(1)$, we have Equation~(\ref{eq:mainAlgStat1}). Lemma~(\ref{lemmaStochastic}), conditional on $b$, results in Equation~(\ref{eq:mainAlgStat2}), with
	\begin{align}
	\label{eq:mainAlgStat1}
	\K(a|b)+f(a) &\lel s + \Ks(f|b),\\
	\label{eq:mainAlgStat2}
	\K(a|b)+f(a) &\lel s + \I(f;\mathcal{H}|b),\\
	\label{eq:mainAlgStat25}
	\K(a|b)+f(a) &\lel s + \K(f|b) -\K(f|b,\mathcal{H}).
	\end{align}
	Using the fact that $\K(a)\lea \K(a|b)+\K(b)$, we get $\K(a)-\K(b)\lea \K(a|b)$, and combined with Equation~(\ref{eq:mainAlgStat25}), we get Equation~(\ref{eq:mainAlgStat3}). Equation~(\ref{eq:mainAlgStat4}) is due to the chain rule $\K(b)+\K(f|b)\lel \K(f)+\K(b|f)$. Equation~(\ref{eq:mainAlgStat5}) follows from the inequality $\K(f|\mathcal{H})\lea \K(f|b,\mathcal{H})+\K(b|\mathcal{H})$. 
	\begin{align}
	\label{eq:mainAlgStat3}
	\K(a)+f(a) &\lel s +\K(b) +\K(f|b) -\K(f|b,\mathcal{H}),\\
	\label{eq:mainAlgStat4}
	\K(a)+f(a) &\lel s +\K(f) +\K(b|f) -\K(f|b,\mathcal{H}),\\
	\label{eq:mainAlgStat5}
	\K(a)+f(a) &\lel s +\K(f) +\K(b|f) -\K(f|\mathcal{H})+\K(b|\mathcal{H}),\\
	\label{eq:mainAlgStat6}
	\K(a)+f(a) &\lel s +\I(f;\mathcal{H}) +(\K(b|f)+\K(b|\mathcal{H})).
	\end{align}
	The remaining part of the proof shows that $\K(b|f)+\K(b|\mathcal{H}) = O(\log (s+\K(b)))$. This is sufficient to proof the theorem due to its logarithmic precision and by the right hand side of the inequality of Equation~(\ref{eq:mainAlgStat3}) being larger than $s+\K(b)$ (up to a logarithmic factor). Since $b$ is a prefix of border, due to proposition~(\ref{prp:borderprefix}), one gets that $\K(b|\mathcal{H})< O(\K(\|b\|))< O(\log \|b\|) < O(\log \K(b))$. Thus combined with Equation~(\ref{eq:mainAlgStat6}) and also Equation~(\ref{eq:mainAlgStat3}), one gets
	\begin{align}
	\K(a)+f(a) &\lel s +\I(f;\mathcal{H}) +\K(b|f).
	\end{align}
	We now prove $\K(b|f)\,{\lea}\,\K(s,\|b\|)$. This follows from the existence of an algorithm, that when given $f$, $s$, and $\|b\|$,  computes $S(b')$ for all $b'\in\{0,1\}^{\|b\|}$ ordered by $\lhd$, and then outputs the first $b'$ such that $S(b')\,{<}\,s$. This output is $b$ otherwise there exists total $b'\,{\lhd}\, b$, with $\|b'\|\,{=}\,\|b\|$, and  $S(b')\,{<}\,s$. This implies the existence of total string ${b'}^-$ such that $S({b'}^-)\,{<}\,s$. This contradicts the definition of $b$ being the shortest total string with $S(b)\,{<}\,s$. So  $\K(b|f)\,{\lea}\,\K(s,\|b\|)$ and thus one gets the final form of the theorem, as shown below. Equation~(\ref{eq:mainAlgStat7}) is again due to the right hand side of Equation~(\ref{eq:mainAlgStat3}).
	\begin{align}
	\nonumber
	\K(a)+f(a) &\lel s +\I(f;\mathcal{H})+\K(s,\|b\|),\\
	\K(a)+f(a) &\lel s +\I(f;\mathcal{H}),\label{eq:mainAlgStat7}\\
	\nonumber
	\min_{a\in\dom(f)}\K(a)+f(a) &\lel -\log \sum_{a\in\dom(f)}\m(a)2^{-f(a)} +\I(f;\mathcal{H}).\end{align}
\end{proof}

\begin{proposition}
	\label{prp:nongrowthStnd}
	For all $\alpha\in\IS$, $x\in\FS$, $\I(x\,{;}\,\mathcal{H})\lea\I(\alpha\,{:}\,\mathcal{H})+\K(x|\alpha)$.
\end{proposition}

\section{Quantum}
\label{sec:quantum}
For a complex number $c=a\,{+}\,b\i$, its conjugate is represented by $c^*=a\,{-}\,b\i$, and we say that $|c|^2 = c^*c = a^2+b^2$. A complex number $c$ is elementary if its real and imaginary components are rational, and its encoding is denoted by $\langle c\rangle = \langle a,b\rangle$. We deal with a finite number $n$ qubits, and we use the Hilbert space $\mathbb{H}_n=(\mathbb{C}^2)^{\otimes n}$. This is a complex vector space of dimension $2^n$, where the scalar field is the set of complex numbers, $\C$. Any vector is a linear combination of the orthonormal basis set $\{\ket{\beta_i}\}$. Pure quantum states $\kpsi$ are represented as unit vectors in this space, with $\kpsi = \sum_i c_i\ket{\beta_i}$ and $\sum_i |c_i|^2=1$. For a pure state $\ktheta$ of $m\,{\leq}\,n$ qubits, $\ket{\theta0{..}}$ is used to denote $\ket{\theta}\ket{0^{n-m}}$. The bra $\bra{\psi}$, is a unit vector of the dual space of $\mathbb{H}_n$.

Mixed quantum states are represented as a probability distribution $\{p_i\}$ over a set of pure states, $\{\ket{\psi_i}\}$, where $\sum_i p_i = 1$ and for each $i$, $p_i\geq 0$. A Hermitian matrix $A$, is a square $2^n\times 2^n$ matrix that is equal to its own conjugate transpose,  where $a_{ij} = a^*_{ji}$. Unless otherwise stated, the dimenions of the matrices in this section are $2^n\times 2^n$. The conjugate transpose of a matrix is denoted by $A^*$.  The eigenvalues of a Hermitian matrix are always real. A Hermitian matrix $A$ is called nonnegative, $A\succeq 0$, if all its eigenvalues are nonnegative. 

For Hermitian matrices $A$ and $B$, we say $A\succeq B$ iff $(A-B)\succeq 0$. For functions $f$ whose range are Hermitian matrices, we use ${\lem}f$ and ${\gem f}$ to denote $\preceq f/O(1)$ and $\succeq f/O(1)$. We use ${\eqm f}$ to denote ${\lem}f$ and ${\gem f}$. 

The trace of a matrix $A$ is represented as $\tr A = \sum_i \bra{\beta_i}A\ket{\beta_i}$ where $\{\ket{\beta_i}\}$ are the basis vectors. Density matrices are nonnegative Hermitian matrices of trace equal to $1$.  Pure states $\kpsi$ have a dual representation as density matrices $\ket{\psi}\bra{\psi}$. Mixed states $\{p_i\}$ over a set of pure states $\{\ket{\psi_i}\}$ have a dual representation as density matrices, as $\sum_i p_i\ket{\psi_i}\bra{\psi_i}$.  

A semi density matrix is a nonnegative Hermitian matrix of trace less than or equal to $1$. Pure states are represented using the notation $\kpsi$, and (semi)-density matrices are represented using the notation $\rho$. 
A pure state $\kpsi=\sum_i c_i\ket{\beta_i}$ of $n$ qubits is elementary, if all its coefficients $c_i$ are elementary. Such elementary states admit a natural encoding $\langle\kpsi\rangle = \langle\{c_1,\dots,c_{2^n}\}\rangle$. Let $\mathcal{P}_n$ be the set of $n$-qubit pure quantum states. For such elementary states, their algorithmic probability is $\m(\ket{\psi})=\m(\langle \kpsi\rangle)$, and their Kolmogorov complexity is $\K(\ket{\psi})=\K(\langle\ket{\psi}\rangle)$.

The identity matrix is denoted by $I$. A matrix $V$ is elementary iff all its entries $v_{ij}$ are elementary. Elementary matrices admit a natural encoding, $\langle V\rangle = \langle \{v_{ij}\,{:}\,i,j\;{\in}\;[1,2^n]\}\rangle$. Let $\mathcal{M}_n$ be the set of elementary semi-density matrices of dimension $2^n$. For $\rho\in\mathcal{M}_n$, $\m(\rho)=\m(\langle\rho\rangle)$.

Transformations between states occur due to unitary matrices $V$ of dimension $2^n$. A matrix $V$ is unitary iff $V^*V=VV^*=I$. A semi-density matrix $\rho$ is lower computable, if there is a program $p$ to the universal Turing machine $U$ that outputs an infinite sequence $\{\langle \rho_i\rangle\}$ of encoded elementary semi-density matrices where $\rho_{i+1}\succeq \rho_i$ and $\lim_{i\rightarrow\infty}\rho_i=\rho$. We say that $p$ lower computes $\rho$.

The universal lower computable semi-density matrix [\citealt*{Gacs01}] is $\bmu$, where if $p$ lower computes $n$ qubit semi-density matrix $\rho$, then $\bmu\gem\m(p|n)\rho$. Furthermore, $\bmu\eqm \sum_{\rho\in\mathcal{M}_n}\m(\rho|n)\rho\eqm\sum_{\ket{\psi}\in\mathcal{P}_n}\m(\ket{\psi}|n)\ket{\psi}\bra{\psi}$.

We now describe an (infinite) encoding scheme for an arbitrary (not necessarily elementary) quantum pure state $\kpsi$. This scheme is defined as an injection between the set of pure states and $\IS$.  We define $
\langle\kpsi\rangle$ to be an ordered list of the encoded tuples $\langle \langle\ktheta\rangle,q,[|\braket{\psi}{\theta}|^2\,{\geq}\,q]\rangle$, over all elementary states $\ktheta$ and rational distances $q\,{\in}\,\Q_{>0}$. For (semi)density matrix $\sigma$, with eigenvectors $\ket{\theta_i}$ and eigenvalues $\delta_i$, we have  
$\langle\sigma\rangle=\langle \langle\ket{\theta_1}\rangle,\langle\delta_1\rangle,\dots,\langle\ket{\theta_n}\rangle,\langle\delta_n\rangle\rangle$.

\section{Vit\'{a}nyi and G\'{a}cs  Complexities}

The G\'{a}cs entropy of a mixed state $\sigma$, is denoted to be $\Hg(\sigma)=\ceil{-\log\Tr\bmu\sigma}$, introduced in [\citealt*{Gacs01}]. The G\'{a}cs entropy of a pure state $\ket{\psi}$ is $\Hg(\ket{\psi})=\Hg(\ket{\psi}\bra{\psi})$. The Vit\'{a}nyi entropy of a $n$ qubit pure state $\ket{\psi}$ is the classical encoding of an approximate pure state plus the fidelity error term with the original state, with $\Hv(\ket{\psi})= \min_{\ktheta\in\mathcal{P}_n}\K(\ktheta|n)-\log |\braket{\psi}{\theta}|^2$, introduced in [\citealt*{Vitanyi99}]. Theorem 8 of [\citealt*{Gacs01}] establishes that $\Hg(\ket{\psi})\lea\Hv(\ket{\psi})\lel 4\Hg(\ket{\psi})$. The following lemma is used to prove Theorem \ref{thr:hvhg}, which provides another bound between $\Hv$ and $\Hg$. This bound shows that quantum states $\ket{\psi}$ that have $\Hv(\ket{\psi}) \gg \Hg(\ket{\psi})$ have large mutual information with the halting sequence $\ch$. 

\begin{lemma}
\label{lmm:qct} For pure quantum state $\ket{\psi}$, $ $\\$ \min_{\ket{\phi}}\K(\ket{\phi})-\log|\braket{\psi}{\phi}|^2 \lel -\log\sum_{\ket{\phi}}\m(\ket{\phi})|\braket{\psi}{\phi}|^2+\I(\langle\kpsi\rangle:\mathcal{H})$.
\end{lemma}

\begin{proof}
	Let $\mathcal{D}$ be a finite set of elementary pure states, computable from $\langle\kpsi\rangle$ and the value $g=\ceil{-\log\sum_{\ket{\phi}}\m(\ket{\phi})|\braket{\psi}{\phi}|^2}$ such that $-\log\sum_{\ktheta\in \mathcal{D}}\m(\ktheta)|\braket{\psi}{\theta}|^2\,{\leq}\, g\,{+}\,1$.  It is computable because there exists an algorithm that can find $\mathcal{D}$ by the following method. The algorithm enumerates  all elementary states $\ktheta$. This algorithm approximates the algorithmic probabilities $\m(\ktheta)$ (from below) with $\widehat{\mathbf{m}}(\ktheta)$. This algorithm uses $\langle\kpsi\rangle$ to approximate $|\braket{\theta}{\psi}|^2$ from below with $\widehat{|\braket{\theta}{\psi}|}^2$. This algorithm stops when it finds a finite set $\mathcal{D}$ such that $-\log \sum_{\ktheta\in\mathcal{D}}\widehat{\mathbf{m}}(\ktheta)\widehat{|\braket{\theta}{\psi}|}^2\leq g+1$. Thus we have that $\K(\mathcal{D}|g,\langle\kpsi\rangle)=O(1)$. Let $f\,{:}\,\mathcal{D}\,{\rightarrow}\,\W$ be a elementary function such that $|{-}\log |\braket{\psi}{\theta}|^2-f(\ktheta)|\leq 1$. One such $f$ is computable relative to $\langle\kpsi\rangle$, and $g$. Firstly this is because $D$ is computable from $\langle\kpsi\rangle$ and $g$. The individual values of $f$ are computable from $\langle\kpsi\rangle$, since $|\braket{\psi}{\theta}|^2$ can be computed to any degree of accuracy. So  $\K(f|g,\langle\kpsi\rangle)=O(1)$ and $-\log\sum_{\ktheta\in \mathcal{D}}\m(\ktheta)2^{-f(\ktheta)} \leq g\,{+}\,2$. One then has that
	\begin{align}
	\nonumber
	\min_{\ket{\phi}}\K(\ket{\phi})-\log|\braket{\psi}{\phi}|^2 &\lea \min_{\theta\in \mathcal{D}}\K(\ktheta)+f(\ktheta)\\
	\label{eqn:sq2}
	&\lel -\log\sum_{\ktheta\in  \mathcal{D}}\m(\ktheta)2^{-f(\ktheta)} +\I(\langle f\rangle;\mathcal{H}).\\
	\label{eqn:sq3}
	&\lel g+\I(\langle f\rangle;\mathcal{H})\\
	\label{eqn:sq4}
	&\lel g+\I(\langle\kpsi\rangle\,{:}\,\mathcal{H})+\K(\langle f\rangle|\langle\kpsi\rangle)\\
	\nonumber
	&\lel g+\I(\langle\kpsi\rangle\,{:}\,\mathcal{H})+\K(g)\\
	\nonumber
	&\lel -\log\sum_{\ket{\phi}}\m(\ket{\phi})|\braket{\psi}{\phi}|^2+\I(\langle\kpsi\rangle\,{:}\,\mathcal{H}).
	\end{align}
	Inequality~\ref{eqn:sq2} is due to Theorem~\ref{thr:MainTheoremAlgStats}. Inequality~\ref{eqn:sq3} is due to the definition of $f$ and $ \mathcal{D}$. Inequality~\ref{eqn:sq4} is due to Proposition~\ref{prp:nongrowthStnd}.
\end{proof}

\begin{theorem}
	\label{thr:hvhg}
	$\Hg(\ket{\psi})\lea\Hv(\ket{\psi})\lel\Hg(\ket{\psi})+\I(\ket{\psi}:\ch|n)$.
\end{theorem}
\begin{proof}
	This follows directly from Lemma \ref{lmm:qct}, relativized to $n$, and the fact that $\Hg(\ket{\psi})\eqa-\log \Tr\bmu\ket{\psi}\bra{\psi}\eqa -\log\sum_{\ket{\phi}}\m(\ket{\phi}|n)|\braket{\phi}{\psi}|^2$.
\end{proof}

\section{Quantum Kolmogorov Complexity}
Quantum Kolmogorov complexity, defined in [\citealt*{BerthiaumeVaLa01}], uses a universal quantum Turing machine to define the complexity of a pure quantum state. The input and output tape of this machine consists of symbols of the type 0, 1, and $\#$. The input is an ensemble $\{p_i\}$ of pure states $\ket{\psi_i}$ of the same length $n$, where $p_i\geq 0$ and $\sum_ip_i=1$. This ensemble can be represented as a mixed state of $n$ qubits. If, during the operation of the quantum Turing machine, all computational branches halt at a time $t$, then the contents on the output tape are considered the output of the quantum Turing machine. The quantum Kolmogorov complexity of a pure state, $\QC[\epsilon](\ket{\psi})$ is the size of the smallest (possibly mixed state) input to a universal quantum Turing machine such that fidelity between the output and $\ket{\psi}$ is at least $\epsilon$. The fidelity between a mixed state output $\sigma$ and $\ket{\psi}$ is $\bra{\psi}\sigma\ket{\psi}$. We require that universal quantum Turing machine be conditioned on the number of qubits $n$, on a classical auxillary tape.

Theorem \ref{thm:qclower} proves that quantum states with low quantum Kolmogorov complexity have low G\'{a}cs entropy. It is related to Theorem 9 in~[\citealt*{Gacs01}] which bounds the larger quantity $\bra{\psi}-\log \bmu\ket{\psi}$ in terms of $\QC[\epsilon](\ket{\psi})$ for $\epsilon>1/2$, but with a weaker error $2(1-\epsilon)n$.
\begin{theorem}
	For pure n-qubit quantum state $\ket{\psi}$, $\Hg(\ket{\psi})\lea \QC[\epsilon](\ket{\psi})+\K(\QC[\epsilon](\ket{\psi})|n)-\log \epsilon$.
	\label{thm:qclower}
\end{theorem}

\begin{proof}
	For each $k$ and $t$ in $\N$, let $\mathcal{H}_{k,t}$ be the smallest linear subspace spanning elementary $k$-qubit inputs to the universal quantum Turing machine $M$ of size $k$ that halt in $t$ steps, outputing an $n$ qubit mixed state. As shown in [\citealt*{Muller08}], if $t\neq t'$, then $\mathcal{H}_{k,t}$ is perpendicular to $\mathcal{H}_{k,t'}$. Let $P_{k,t}$ be the projection onto $\mathcal{H}_{k,t}$. For each $k$ and $t$, the universal quantum Turing machine defines a completely positive map $\:\Psi_{k,t}$ over $\mathcal{H}_{k,t}$, where $\Psi_{k,t}(\nu)=\sigma$ implies that the quantum Turing machine, with semi-density matrix $\nu$ of length $k$ on the input tape will output the $n$ qubit semi-density matrix $\sigma$ and halt in time $t$. Let $\rho$ be a $k$ qubit mixed state that minimizes $k=\QC[\epsilon](\ket{\psi})$ in time $t$. 
	\begin{align*}
	\rho&\leq P_{k,t}\\
	2^{-k}\rho&\leq 2^{-k}P_{k,t}\\
	\Psi_{k,t}2^{-k}\rho&\leq \Psi_{k,t}2^{-k}P_{k,t}\\
	\Psi_{k,t}2^{-k}\rho&\leq \sum_t\Psi_{k,t}2^{-k}P_{k,t}
	\end{align*}
	The semi density matrix  $\sum_t\Psi_{k,t}2^{-k}P_{k,t}$ is lower computable relative to $k$ and $n$, so
	\begin{align*}
	\m(k|n)2^{-k}\Psi_{k,t}\rho&\leq \m(k|n)\sum_t\Psi_{k,t}2^{-k}P_{k,t}\lem \bmu\\
	\m(k|n)2^{-k}\bra{\psi}\Psi_{k,t}(\rho)\ket{\psi}&\lem \bra{\psi}\bmu\ket{\psi}\\
	k+\K(k|n)-\log \epsilon&\gea \Hg(\ket{\psi}).
	\end{align*}
\end{proof}

\begin{lemma}
	\label{thr:mixedstatecoding}
	For each density matrix $\sigma$ with $k=\ceil{-\log\sum_{\rho\in\mathcal{M}_n}\m(\rho|n)\Tr\rho\sigma}$, there is an elementary density matrix $\rho$ such that
	\begin{enumerate}
		\item $\K(\rho|n)-\log\Tr\rho\sigma \lel k$,
		\item $\K(\rho|n)\lel \I(\langle\sigma\rangle:\mathcal{H}|n)+O(\log k)$.
	\end{enumerate}
\end{lemma}

\begin{proof}
	By definition of $\bmu$, we have that $k\eqa-\log\Tr\bmu\sigma$. For a given sequence $x$, we define the following semi-density matrix $\nu[x]=\sum \{2^{-\|y\|}\ket{\phi}\bra{\phi}\,{:}\,x\sqsubseteq y, U_n(y)=\ket{\phi}\textrm{ has $n$ qubits}\}$. Thus $\Tr \nu[x]\leq 2^{-\|x\|}$. For $i\in\mathbb{N}$,  if $\mathcal{B}[i]=0$, let $\nu[i]$ be a 0 matrix, otherwise if $\mathcal{B}[i]=1$, then $\nu[i] = \nu[\mathcal{B}[0..i-1]0]$. This definition makes sense because $\mathcal{B}[0..i{-}1]0$ is total whenever $\mathcal{B}[i]=1$. A visual description of $\nu[i]$ can be seen in Figure \ref{fig:vi}.
	
	\begin{figure}[h!]
		\begin{center}
			\includegraphics[width=0.7\columnwidth]{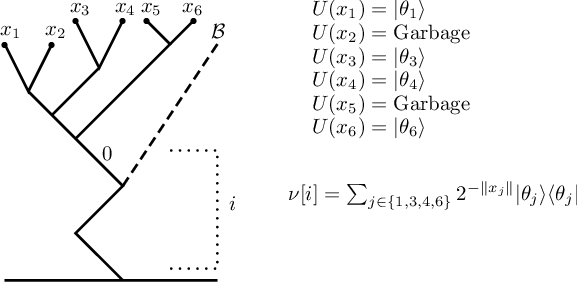}
			\caption{The above figure shows an example $\nu[i]$, where $\mathcal{B}[i]=1$. Two extensions of $\mathcal{B}[0..i-1]0$ do not produce elementary pure states when given as input to the universal Turing machine $U$. They are $x_2$ and $x_5$. The semi density matrix $\nu[i]$ has $\Tr \nu[i]\leq 2^{-i}$ and is the weighted sum of elementary density matrices $\ket{\theta_i}\bra{\theta_i}$.}
			\label{fig:vi}
		\end{center}
	\end{figure}

So $2^{-k}\eqm\Tr\bmu\sigma\eqm\sum_{i=1}^\infty \Tr \nu[i]\sigma$. Since the trace of $\nu[i]$ is the  weighted sum of a finite number of prefix free extensions of a string of length $i$, (or 0), $\Tr\nu[i]\leq 2^{-i}$. So $2^{-k}\eqm \sum_{i=1}^{k+1}\Tr\nu[i]\sigma+\sum_{i=k+2}^\infty\Tr\nu[i]\sigma \lem \sum_{i=1}^{k+1}\Tr\nu[i]\sigma+2^{-k-1}$. So $\sum_{i=1}^{k+1}\Tr\nu[i]\sigma\gem 2^{-k-1}$. So there is an $j$, $1\leq j\leq k+1$ such that $\Tr\nu[j]\sigma> c2^{-k-1}/(k+1)$, where $c$ is a positive rational constant solely dependent on the universal Turing machine $U$.  Let $j$ be the smallest index where this occurs. 
	
We define $\rho$ to be equal to $\nu[j]/\Tr\nu[j]$. $\rho$ is an elementary density matrix, which is computable from the first $j-1$ bits of the border sequence, with $\K(\rho|n)\lea \K(\mathcal{B}[0..j-1]|n) \lea j+2\log k$.  Since $\Tr\nu[j]\leq 2^{-j}$, we have that $-\log \Tr\rho\sigma \lea k-j+\log k$. So $\K(\rho)-\log\Tr\rho\sigma\lea k+3\log k$, proving (1).
	
It remains to prove that $\K(\mathcal{B}[0..j-1]|n) \lel \I(\langle \sigma\rangle:\mathcal{H}|n)+\K(k|n)$. By the definition of $\I$, for 
\begin{align*}
\I(\mathcal{H}:\langle\sigma\rangle|n)&=\log \sum_{x,y\in \FS}2^{\I(x:y|n)}\m(x|\mathcal{H},n)\m(y|\langle\sigma\rangle,n). \\
&>\log \sum_{x\in \FS}2^{\I(x:x|n)}\m(x|\mathcal{H},n)\m(x|\langle\sigma\rangle,n)\\
&\gea\log 2^{\K(\mathcal{B}[0..j-1]|n)}\m(\mathcal{B}[0..j-1]|\mathcal{H},n)\m(\mathcal{B}[0..j-1]|\langle\sigma\rangle,n)\\
&\eqa \K(\mathcal{B}[0..j-1]|n) - \K(\mathcal{B}[0..j-1]\,{|}\,\langle \sigma\rangle,n) - \K(\mathcal{B}[0..j-1]\,{|}\,\mathcal{H},n).
\end{align*}
Since the border sequence is computable from the halting sequence $\K(\mathcal{B}[0..j-1]\,{|}\,\mathcal{H}) \lea \K(j)<O(\log k)$, therefore it is sufficient to prove $\K(\mathcal{B}[0..j-1]\,{|}\,\langle\sigma\rangle,n)\lea\K(k,j|n)<O(\log k)$.
	
Given $k$, $j$, $\langle \sigma\rangle$, and $n$, there is a program that can enumerate $\nu[x]$ over all total strings $x$ of length $j$. This program also uses $\langle \sigma\rangle$ to compute from below $\Tr \nu[x]\sigma$, until it finds a string $y$ such that $\Tr\nu[y]\sigma> c2^{-k-1}/(k+1)$. Note that there is a unique string $y$ with this property and $y=\mathcal{B}[0..j-1]0$. Otherwise $y=\mathcal{B}[0..h-1]0z$, for some $h<j$, $\|h\|+\|z\|=j$, with $\mathcal{B}[h]=1$. However, then $\Tr \nu[h] > c2^{-k-1}/(k+1)$, causing a contradiction for the definition of $j$. Thus the program can find $\mathcal{B}[0..j-1]$, completing the proof for (2).
\end{proof}

Theorem \ref{thr:qchg} proves that quantum states with very low G\'{a}cs complexity, $\Hg$, will be easy to approximate with a small input to a universal quantum Turing machine. This addresses open question (1) of \citealt*{Gacs01}. Let $\mathcal{Q}_n$ be the set of $n$-qubit pure states. For a sequence $\ket{\psi_n}\in\mathcal{Q}_n$ of quantum states where $\Hg(\ket{\psi_n})=O(1)$, the following theorem shows that $\QC[\Omega(1)](\ket{\psi_n})=O(1)$.\newpage

\begin{theorem}
	\label{thr:qchg}
	$\QC[\bra{\psi}\bmu\ket{\psi}/\Hg(\ket{\psi})^{O(1)}](\ket{\psi})\lel \Hg(\ket{\psi})$.
\end{theorem}
\begin{proof}
	We use reasoning from Theorem 7 in \citealt*{Gacs01}. From Lemma\begin{flushright}
		Theorem
	\end{flushright} \ref{thr:mixedstatecoding}, there exists a $\rho$ such that $\K(\rho|n)-\log\bra{\psi}\rho\ket{\psi}\lel \Hg(\ket{\psi})$. Let $\ceil{-\log \bra{\psi}\rho\ket{\psi}}=m$. Let $\ket{u_1},\ket{u_2},\ket{u_3},\dots$ be the eigenvectors of $\rho$ with corresponding eigenvalues $\gamma_1\geq \gamma_2 \geq \gamma_3\dots$ For $y\in\N$, let $\rho_y=\sum_{i=1}^y\gamma_i\ket{u_i}\bra{u_i}$. We expand $\ket{\psi}$ in the basis of $\{\ket{u_i}\}$ with $\ket{\psi}=\sum_i c_i\ket{u_i}$. So we have that $\sum_{i}\gamma_i|c_{i}|^2\geq 2^{-m}$. Let $s\in\N$ be the first index $i$ with $\gamma_i<2^{-m-1}$. Since $\sum_i \gamma_i\leq 1$, it must be that $s\leq 2^{m+2}$. So
	\begin{align*}
	\sum_{i\geq s}\gamma_i|c_{i}|^2&<2^{-m-1}\sum_{i}|c_{i}|^{2}\leq 2^{-m-1}, \\
	\bra{\psi}\rho_{2^{m+2}}\ket{\psi}\geq\Tr \bra{\psi}\rho_s\ket{\psi} &> \sum_{i< s}\gamma_i|c_{i}|^2 \geq 2^{-m}-\sum_{i\geq s}u_i|c_{i}|^2 > 2^{-m-1}.
	\end{align*}
	We now describe a program to the universal quantum Turing machine (with auxilliary tape containing $\langle n\rangle$) that will construct $\rho_{2^{m+2}}$. The input is an ensemble $\{\gamma_i\}_{i=1}^{2^{m+2}}$ of vectors $\{\ket{cB(i)}\}$, where $B(i)$ is the binary encoding of index $i\in\N$ which is of length $m+2$. Helper code $c$ of size $\eqa\K(p|n)$ transforms each $\ket{cB(i)}$ into $\ket{u_i}$. Thus the size of the input is $\lea \K(p|n)+m\lel\Hg(\ket{\psi})$. The fidelity of the approximation is $\bra{\psi}\rho_{2^{m+2}}\ket{\psi} > 2^{-m-1} \geq 2^{-\Hg(\ket{\psi})-O(\log\Hg(\ket{\psi}))}\geq \bra{\psi}\bmu\ket{\psi}/\Hg(\ket{\psi})^{O(1)}$.
\end{proof}


\end{document}